\documentclass[12pt,reqno]{amsart}
\usepackage[left=2.9cm,top=2.9cm,right=2.9cm,bottom=2.9cm]{geometry}

\usepackage{amssymb,amsfonts,latexsym,amstext,epsfig,amsmath}

\newtheorem{theorem}{Theorem}
\newtheorem{algorithm}[theorem]{Algorithm}

\newtheorem{definition}[theorem]{Definition}

\newtheorem{lemma}[theorem]{Lemma}

\usepackage{subfigure}
\newcommand{\flfrac}[2]{\left \lfloor \frac{#1}{#2} \right \rfloor}

\usepackage{url}

\usepackage{appendix}

\begin{document}

\title{Finding safe strategies\\
 for competitive diffusion on trees}
\author{Jeannette Janssen \and Celeste Vautour}

\begin{abstract}
We study the two-player safe game of Competitive Diffusion, a game-theoretic model for the diffusion of technologies or influence through a social network. In game theory, safe strategies are mixed strategies with a minimal expected gain against unknown strategies of the opponents. Safe strategies for competitive diffusion lead to maximum spread of influence in the presence of uncertainty about the other players. We study the safe game on two specific classes of trees, spiders and complete trees, and give tight bounds on the minimal expected gain. We then use these results to give an algorithm which suggests a safe strategy for a player on any tree. We test this algorithm on randomly generated trees, and show that it finds strategies that are close to optimal.
\end{abstract}

\maketitle

\section{Introduction}
Online social networks such as {\sl Facebook}, {\sl Twitter} and {\sl LinkedIn} have an increasingly important role in the spread of information through society. News about all kind of topics can spread quickly along the \lq\lq friend'' or ``follower" links in the network. Understanding and modelling this process, and determining best strategies for reaching a large number of users, is instrumental for commercial applications such as {\sl viral marketing}, but also for social activism and societal benefit, such as countering false rumours, spreading information about safe health practices, etc. 

{\sl Competitive Diffusion} is a game-theoretic model for the diffusion of information in a network which was introduced in \cite{Alon2010}. This game is built on the assumption that there are several players, who wish to spread competing information. One can think of companies wishing to encourage consumers to adopt their products, or political organizations wishing to spread a point of view about a contentious issue. The goal of each player is to reach the largest possible number of users. The messages spread by the players are assumed to be competitive, so any user who adopts the view of one of the players will not be susceptible to the messages sent by the other players. Moreover, users adopt the view of the player whose message is the first to reach them. If two competing messages reach the user at the same time, the user adopts a neutral position and effectively blocks the passage of information. 

Competitive diffusion lends itself to analysis via game theory. Because of the possibilities that users turn ``neutral", it is not a zero-sum game. In this paper, we consider the associated {\sl safe game}. This game focusses on one particular player, here referred to as Player 1, and the aim is to maximize the minimal gain of Player 1, regardless of the strategies of the other players. The safe game can be interpreted as the game where all other players have as their goal to minimize the gain of Player 1, rather than maximize their own gain. The reason to adopt the safe game scenario is because the traditional game assumes full information about the strategies of the other player. The safe game explores the scenario where the strategies of the other players are unknown, and thus the safest scenario for Player 1 is to assume that the other players are actively countering her strategy. This contrasts with the analysis of competitive diffusion in terms of pure Nash equilibria, where the assumption is that everyone is fully aware of the strategies of the other players, but the aim is for all players to maximize their own gain. 

Our results concern the safe game of competitive diffusion played on trees. We give an optimal safe strategy for full $q$-ary trees, and give asymptotically optimal safe strategies for spiders. For spiders consisting of a number of paths of equal length joined at a common vertex, we show that the safety value equals the gain of the disadvantaged player in a Nash equilibrium for competitive diffusion. In other words, we cannot improve  on the safe gain by assuming the fully open and self-interested game rather than the adversarial setting.  

Finally, we use results for special types of trees to develop a heuristic algorithm that can be applied to any tree. We show that the algorithm gives optimal results when applied to certain subclasses of trees. We also test the algorithm on randomly generated trees, and show that the safe strategies found by the algorithm have performance that is close to optimal. 

\subsection{Related Work}

The first studies on the spread of influence through social networks assumed a passive model. The goal was to predict how information diffuses through a network starting from a given set of vertices. If the information reaches a vertex, this vertex is said to be {\sl activated}. There are mainly two types of diffusion models, threshold models and cascade models. The difference in these model is in how vertices become activated.  

In threshold models, vertices become activated once a variable associated with the neighbourhood of a vertex  surpasses a certain threshold. The most commonly used is the {\sl Linear Threshold Model} (see \cite{Granovetter1978} and \cite{Kempe2003}). In this model, each vertex $v$ has a threshold $\theta_v$, and a vertex $v$ is influenced by each of its neighbours, $w$, by a weight $b_{v,w}$. A vertex becomes activated once the sum of the weights of its activated neighbours exceeds $\theta_v$.

In cascade models, as a vertex becomes activated, it activates each of its neighbours with a given probability. The most well-known is the Independent Cascade Model (see \cite{Goldenberg2001} and \cite{Kempe2003}). In this model, we also start with an initial set of activated vertices. Here, each edge $vw$ is assigned a probability $p_{v,w}$. If vertex $v$ becomes activated, its neighbour $w$ will become activated in the next round with probability $p_{v,w}$. The spread of influence in competitive diffusion can be seen as a cascade model where the activation probability equals 1.

The optimization problem studied in these diffusion models is how to choose the set of starting vertices so that the expected diffusion is maximized (see \cite{Kempe2003}, \cite{Borodin2010} and \cite{Chen2009}). In other words, the goal is to identify a set of initial influenced users which will bring a greater overall influence throughout the network.  A related approach is through {\sl Voronoi games} on graphs (see \cite{Demaine2011} and \cite{Sayan2013}).
Here the players choose a set of vertices, and all other vertices are assigned to the starting vertex which is closest to it. 

Competitive diffusion, as proposed in  \cite{Alon2010}, is the first game-theoretic model in which the players are considered to be outside the social networks. Players choose initial users to influence and their goal is to reach the most users.  In \cite{Alon2010} (see also erratum \cite{Takehara2012}), the authors discuss the relationship between the diameter of the graph and the existence of pure Nash equilibria. A pure Nash equilibrium is a strategy which corresponds to a set of initial vertices, whereas a mixed strategy represents a probabilistic approach where starting vertices are chosen with a certain probability. In \cite{Small20132}, the existence of a pure Nash equilibrium for competitive diffusion on trees is shown, while in \cite{Roshanbin2014}, results on pure Nash equilibria are given for several classes of graphs. Moreover, \cite{Small2013} considers the competitive diffusion on a recently proposed model for on-line social networks and discusses the existence of Nash equilibria. The safe game for competitive diffusion was introduced in \cite{Boudreau2013}, and some results for paths were given.

Generalizations of competitive diffusion were proposed in \cite{Goyal2012}  and \cite{ Goldberg2012}.  In \cite{Goyal2012}, the agents choose an allocation of budgeted seeds over the vertices and the diffusion process is stochastic. In \cite{ Goldberg2012}, the agents choose an initial set of vertices and the diffusion is a threshold model.

\section{Preliminaries}

\subsection{Competitive Diffusion Model}
Let us start by recalling the model Competitive Diffusion from \cite{Alon2010}.
Let $G$ be a graph with $n$ vertices and suppose there are $p$ players, $P_1,...,P_p$ each having a distinct assigned colour (not white or grey). The strategy of each player is to choose a vertex in $G$ as their starting vertex. The game begins by colouring each of the starting vertices of the players and then proceeds with the diffusion of colours through $G$ as follows: at each wave of diffusion, a vertex that has one or more neighbours with a certain colour inherits that colour while a vertex that has two neighbours with different colors turns grey.  
The diffusion finishes when all the vertices have either inherited a colour, have turned grey or are forced to stay uncoloured (white) being blocked off by grey vertices. In the end, the gain of each players is the number of vertices that has assumed his or her colour. The winner of the game is the player that has the greatest gain.
We note that if two or more players have the same starting vertex, then this vertex immediately turns grey.

While the game can be played with any finite number of players, this paper concentrates on the two-player version of the game. In the following, the two players will be called Player 1 (She) and Player 2 (He).

\subsection{Mixed strategies}

Consider competitive diffusion on an undirected graph $G$ with vertex set $V(G)=\{v_1,v_2,...,v_n\}$. We will denote the game matrix of Player 1 by $A_G$. Precisely, this is the matrix so that the entry $(A_G)_{ij}$ gives Player 1's gain if she chooses starting vertex $v_i$ and Player 2 chooses starting vertex $v_j$. A {\sl mixed strategy} for a player is a vector $(x_1,x_2,...,x_n)$ so that $\sum_{i=1}^n x_i=1$ and $x_i\geq 0$ for $i=1,\dots ,n$. It should be interpreted as a probabilistic strategy, where $x_i$ and $y_i$ is the probability the player  chooses vertex $v_i $ as starting vertex. 
Accordingly, the expected gain of Player 1 when she plays the mixed strategy $X=(x_1,x_2,...,x_n)$ and Player 2 plays the mixed strategy $Y=(y_1,y_2,...,y_n)$ is
\begin{equation}
\label{CompetitiveDiffusionExpectedGain}
Gain(G,X,Y)=XA_GY^T.
\end{equation}

Let $S_n=\{(z_1,z_2,...,z_n) \mid z_i\geq 0, 1 \leq i \leq n, \sum_{i=1}^nz_i=1\}$ be the strategy set of the players. We will use the special notation $Z(v_k)$ for a mixed strategy equivalent to a pure strategy, i.e.~when a player chooses vertex $v_k \in V(G)$ with probability 1 and the other vertices with probability 0. Precisely, $Z(v_k)=(z_1,z_2,...,z_n)$ with
\begin{equation}
z_i=
\begin{cases}
1 &  \text{if } i=k \\
0 &  \text{otherwise.}
\end{cases}
\end{equation}

\subsection{Safe Game and Notations}

The {\sl safe game} for competitive diffusion is the zero-sum game where the game matrix is $A_G$, which is the game matrix for Player 1 in competitive diffusion. The {\em safety value} for Player 1 is 
\[
value(A_G)=\min_{Y\in S_n}\max_{X\in S_n} XA_GY^T=\max_{X\in S_n} \min_{Y\in S_n}XA_GY^T.
\]
Moreover, if $Gain(G,X^*,Y^*)=value(A_G)$, then $X^*$ is called the {\em maxmin strategy} for Player 1, and $Y^*$ is called the {\em minmax strategy} for Player 2.

Any (mixed) strategy for Player 1 in the safe game will be referred to as a {\sl safe strategy}.  Correspondingly, any mixed strategy for Player 2 in the safe game will be called an {\sl opposing strategy}. In this paper, all strategies are assumed to refer to the safe game, unless stated otherwise.

The {\em guaranteed gain} of Player 1 with the safe strategy $X$, $GGain(G,X)$, is the minimal gain that Player 1 could receive with the strategy $X$, i.e.

\begin{equation}
\label{GuaranteedGain}
GGain(G,X)=\min_{Y \in S_n} XA_GY^T=\min_{y \in V(G)} Gain(G,X,Z(y)).
\end{equation}
The {\em maximal gain} of Player 1 against the opposing strategy $Y$ of Player 2 is the maximal gain that Player 1 could receive when Player 2 chooses the strategy $Y$, i.e.
\begin{equation}
MGain(G,Y)=\max_{X \in S_n} XA_GY^T=\max_{v_i\in V(G)} Gain(G,Z(v_i),Y).
\label{MaximalGain}
\end{equation}

Note that the guaranteed gain of Player 1 with any pure strategy $Z(v_k)$ equals zero, since Player 2 can counter by playing the same strategy $Z(v_k)$, reducing the gain of Player 1 to zero. Thus, any optimal safe strategy will be mixed. 

The guaranteed gain with any safe strategy $X$ for Player 1 is a lower bound on the safety value while the maximal gain of Player 1 against any opposing strategy $Y$ of Player 2 is an upper bound on the safety value. Mathematically, we have 
\begin{equation}
GGain(G,X) \leq value(A_G) \leq MGain(G,Y).
\end{equation}

Thus, any pair of strategies for Player 1 and Player 2 give a lower and upper bound on the safe gain $value(A_G)$. In the following sections, we find strategies so that these bounds are tight or asymptotically tight. 

\subsection{Trees: Weights and Centroid}
In this paper, we study the safe game for competitive diffusion on trees. Here we introduce some facts about trees which are relevant to our analysis.

There exists more than one notion of center in a graph. We use the``branch weight" notion of centroid from \cite{Mitchell1978}. 
A {\em branch} of a tree $T$ at a vertex $v$ is a maximal sub-tree of $T$ which has $v$ as a leaf.  Correspondingly, the {\em weight} of the vertex $v$, $w(v)$, is the maximum number of edges in any branch of $v$. We also use the notation $\overline{w}(v)=n-w(v)$, where $n$ is the size of $T$. The {\em centroid} of $T$, denoted $C(T)$, is the set of vertices which have the minimal weight in $T$ .

In a tree, it is known that the centroid is either a single vertex or two adjacent vertices \cite{Knuth}. Moreover, a tree which has only one vertex as centroid is called a {\em centroidal tree} and a tree which has two vertices as centroid is called a {\em bicentroidal tree}. We also have the following condition for a vertex to be in the centroid of a tree. 

\begin{theorem}[from \cite{Kang1975}]
\label{ThmCentroidVertex}
Let $T$ be a tree of size $n$ with $k$ branches having $n_1,n_2,...,n_k$ edges, respectively. Let $v$ be a vertex of $T$. Vertex $v$ is a centroid vertex of $T$ if and only if $n_i \leq \frac{n}{2}$ for $1 \leq i \leq k$. 
\end{theorem}

The following lemmas on the weights of vertices will be helpful in establishing the main results.
\begin{lemma}
\label{TreeCentroidinBranch}
For any tree $T$ of size $n$, if $v$ is a vertex of $T$ not part of the centroid $C(T)$, then its weight $w(v)$ is the number of edges in the branch at $v$ in which  $C(T)$ is located.
\end{lemma}

\begin{proof}
By way of contradiction, suppose that $B$, the branch at $v$ in which the centroid is located, is not the branch with the maximum number of edges. Let $c$ be a vertex in the centroid. Since the weight of a vertex is the maximum number of edges in one of its branches, we have $w(c) \geq n- |B|$.
On the other hand, since $B$ is not the branch at $v$ with the maximum number of edges, we have $w(v) \leq n - |B|$.
Thus, $n-|B| \leq w(c) \leq w(v) \leq n-|B|$ since the centroid is the vertex with the minimal weight in $T$. Hence, $w(c)=w(v)$ which is a contradiction since $v$ is not a vertex in the centroid of $T$.
\end{proof}

\begin{lemma}
\label{Lemmanminusweight}
For any tree $T$ of size $n$, if $v$ is a vertex of $T$ not part of the centroid $C(T)$, then
\begin{align}
w(v)>\overline{w}(v).
\end{align}

\end{lemma}
\begin{proof}
If $v$ is not a centroid vertex, at least one of its branches must have more than $\frac{n}{2}$ vertices by Theorem \ref{ThmCentroidVertex}. Since $w(v)$ is the number of edges in the largest branch at $v$, we must have $w(v)>\frac{n}{2}$. Thus $w(v)>n-w(v)$.
\end{proof}

\section{Spiders}

We start the study of the two-player safe game of competitive diffusion by giving tight bounds on the safety value for the game on special cases of trees, spiders and complete trees. The corresponding safe strategies will give insight to suggest a good safe strategy for a player on any tree.

A {\em spider}  is a tree with one and only one vertex of degree exceeding 2. The vertex with degree exceeding 2 is called the {\em body} of the spider. Moreover, any branch at the body of the spider is none other than a non-trivial path and is called a {\em leg} of the spider. (See \cite{Lin2011}).

Let us denote the $m$ legs of a spider $S$ by $\{ s_1,s_2,...,s_m \}$ and their lengths respectively by $ \{ l(s_1),l(s_2),...,l(s_m) \}$. We will label a vertex $v_i$ in $S$ by an ordered pair $(d,s)$ where $d$ is the number of edges from the vertex $v_i$ to the body of the spider and where $s \in \{ 1,2,...,m \}$ is the index of the leg the vertex belongs to. By convention, the body of the spider will be identified by the ordered pair $(0,0)$. We suggest the following safe strategy for Player 1 on a spider with legs of equal lengths. The strategy has positive probabilities of choosing the body of the spider and the first $k$ vertices of the legs.
\begin{definition}
\label{DefStrategyC_S(k)}
Given a spider, $S$, with $m$ legs each having $\ell$ vertices. Let the vertices of $S$ be labelled $v_0,v_1,\dots , v_{m\ell }$, where  $v_0$ is the body of $S$, and for $d \in \{1,2,...,\ell \}$ and $s \in \{1,2,...,m\}$, the vertex $v_i$ where $i=d+(s-1)\ell$  is the vertex labelled by the ordered pair $(d,s)$. For any $k \in \{0,1,...,\ell\}$, define the strategy $C_S(k)$ to be the strategy $(z_0,z_1,z_2,...,z_{m\ell})$ as follows. Consider vertex $v_i$ with label $(d,s)$. Then its probability 
\begin{equation}
z_{i}=
\begin{cases}
0 & \text{ if } k < d \leq \ell, \\
\frac{1}{mk+1} & \text{ if } 0 \leq d \leq k .
\end{cases}
\end{equation}
\end{definition}
Considering the strategy $C_S(k)$ as a safe strategy for Player 1 and as an opposing strategy for Player 2 leads to the following bounds on the safety value.

\begin{theorem}
\label{SpiderS1SafetyValueBounds}
In the two-player Competitive Diffusion on $S$ with $m$ legs each having $\ell$ vertices, the safety value of Player 1 is between $\ell-\frac{\sqrt{\ell }}{\sqrt{m}}+\mathcal{O}(1)$ and $\ell$ (asymptotics as $\ell$ goes to infinity).
\end{theorem}
\begin{proof}
Assume the vertices of the spider are labelled as in the statement of Definition \ref{DefStrategyC_S(k)}, that is, if vertex $v_j$ has label $(d,s)$, then $j=d+(s-1)\ell$.
As a lower bound, we have the guaranteed gain of Player 1 with the strategy $C_S(k)$. 
%\begin{align}
%GGain(S,C_S(k))=\min_{1\leq j\leq n} Gain(S,C_S(k),Z(v_j)).
%\end{align}
As stated in (\ref{GuaranteedGain}), the guaranteed gain is the minimal gain of Player 1 over all the possible starting vertices for Player 2. Due to symmetry, we only need to consider the body of the spider and the vertices on one of the legs of $S$.  The expected gain of Player 1 when Player 2 chooses the body is

\begin{eqnarray*}
Gain(S,C_S(k),Z(v_0))=\frac{1}{mk+1}  \left(0+m \, \sum_{\delta =1}^k \left( \ell -\flfrac{\delta }{2} \right) \right).
\end{eqnarray*}
Here, the summation is over the gain obtained when Player 1 chooses a vertex at distance $\delta$ from the body.
Evaluating the sum, we obtain
\[
Gain(S,C_{S}(k),Z(v_0))=
\begin{cases}
\frac{m}{mk+1} \left( k\ell-\frac{k^2}{4} \right) \text{ if } k \text{ is even} \\
\frac{m}{mk+1} \left( k\ell-\frac{k^2}{4}+\frac{1}{4} \right) \text{ if } k \text{ is odd}.
\end{cases}
\]

If Player 2 chooses a vertex $v_j$ with label $(d,s)$ where $s>0$ and $d>k$ the expected gain of Player 1 is
\begin{eqnarray}
\label{Equation1}
Gain(S,C_{S1}(k),Z(v_j)) &=&\frac{1}{mk+1} 
\left( \flfrac{1+d}{2}+n-\ell-1 \right. \notag\\
&+& \sum_{\delta =1}^{k} \left( \flfrac{\delta +d-1}{2} +n-\ell \right)  \notag \\
&+&(m-1)\,\left.\sum_{\delta=1}^{k} \left( \flfrac{-\delta + d+1}{2}+n-\ell -1 \right) \right).%\notag
\end{eqnarray}
Here, the first summation ranges over the gain obtained when Player 1 chooses a vertex at distance $\delta$ from the body on the same branch as the starting vertex of Player 2, and the second summation does the same for a vertex on a different branch. 
The expression above  increases with $d$. Thus, the minimum expected gain for Player 1 occurs when $d=k+1$. However, the substitution $d=k+1$ in (\ref{Equation1}) gives an expected gain which is greater than $Gain(S,C_S(k),Z(v_0))$. 

Finally, consider the case where Player 2 chooses a vertex $v_j$ with label $(d,s)$ where $s>0$ and $0<d\leq k$. Then the expected gain of Player 1 is

\begin{equation*}
\begin{split}
Gain(S,C_{S1}(k),Z(v_j))=\frac{1}{mk+1} \left( 
 \sum_{\delta =0}^{d-1} \left( \flfrac{\delta + d-1}{2}+n-\ell \right)  + \sum_{\delta =d+1}^{k} \left( \ell -\flfrac{\delta + d}{2} \right)\right.\\
+ (m-1) \left. \left( \sum_{\delta =1}^{ d-1} \left( \flfrac{-\delta + d+1}{2} + n-\ell -1 \right)+\ell  + \sum_{\delta = d+1}^{k} \left(\ell +\flfrac{-\delta + d+1}{2} \right) \right) \right).
\end{split}
\end{equation*}

In the above expression, the first line refers to the gain when Player 1 chooses a starting vertex on the same branch as Player 2, while the second line refers to the case where the starting vertices are on different branches.
We can show that the expected gain in this case is greater than $Gain(S,C_S(k),Z(v_0))$. Thus, $GGain(S,C_S(k))=Gain(S,C_S(k),Z(v_0))$.  

Maximizing $GGain(S,C_S(k))$ over $k$ gives $k^*=\frac{2\sqrt{\ell }}{\sqrt{m}}+\mathcal{O}(1)$ as the optimal integer choice for $k$ and $GGain(S,C_S(k^*))=\ell -\frac{\sqrt{\ell }}{\sqrt{m}}+\mathcal{O}(1)$. 

For the upper bound, we have the maximal gain of Player 1 when Player 2 has the strategy $C_S(k)$ with $k=0$. If $k=0$, the strategy of Player 2 is simply to choose the body of the spider. In this case, the maximal gain Player 1 can  obtain is $\ell $, the number of vertices in one leg.
\end{proof}

\section{Complete $m$-ary Trees}

A {\em complete m-ary tree} $(m \geq 2)$ of height $h$, which we will denote by $T(m,h)$, is a rooted tree in which every internal vertex has exactly $m$ children and all leaves have depth $h$.
The number of vertices in $T(m,h)$ is $n=\frac{m^{h+1}-1}{m-1}$. Let us identify a vertex $v_j$  by an ordered pair $(d,e)$, where $d$ is the depth of $v_j$, and $e$ is the position of the vertex  in levels $d$  if the vertices in the levels are numbered from left to right by $\{0,1,2,...,m^{d}-1\}$. By convention, the root of the tree will be identified by the ordered pair $(0,0)$.

%Note that with this notation, one could determine if a vertex $v_i$ with label $(d,e)$ is a descendant of a vertex $v_j$ with label $(d',e')$ by checking if 
%\begin{equation*}
%e \in \{ e' m^{d-d'}, (e'+1) m^{d-d'}-1\}.
%\end{equation*}

In the following, we will use the notation $Z(d,e)$ to denote the pure strategy $Z(v_j)$ where $v_j$ has label $(d,e)$. 

We suggest the following safe strategy for Player 1 and opposing strategy for Player 2. The strategies have positive probabilities of choosing the root and the vertices in the first level of $T(m,h)$.
\begin{definition}
\label{DefStrategymu1}
Let the {\em strategy $\mu_{1}$} be a mixed strategy $(x_1,x_2,...,x_n)$ on $T(m,h)$ where 
$n=\frac{m^{h+1}-1}{m-1}$ and for all $i\in \{ 1,2,\dots ,n\}$,
\begin{equation*}
x_i=
\begin{cases}
\alpha_1 =\frac{m^h-1}{m^{h+2}-m^{h+1}+m^h-1} & \text{ if }v_i\text{ is the root}, \\
\beta_1 =\frac{(m-1)m^h}{m^{h+2}-m^{h+1}+m^h-1} & \text{ if }v_i\text{ has depth }1,\\
0 & \text{ if }v_i\text{ has depth } d>1.
\end{cases}
\end{equation*}
\end{definition}
The probabilities $\alpha_1$ and $\beta_1$ in the strategy $\mu_1$ were obtained by solving
\begin{align}
\label{mu1alphaandbeta}
\begin{split}
Gain(T(m,h),X_1,Z(0,0))&=Gain(T(m,h),X_1,Z(1,1))\\
\Leftrightarrow \alpha_1 \cdot 0 + m \beta_1 \cdot \left( \frac{m^h-1}{m-1}\right)
&=\alpha_1 \cdot m^h + \beta_1 \cdot 0 + (m-1) \beta_1 \cdot \left( \frac{m^h-1}{m-1} \right),
\end{split}
\end{align}
subject to the condition that $\alpha_1+m\beta_2=1$.

\begin{definition}
\label{DefStrategymu2}
Let the {\em strategy $\mu_2$} be a mixed strategy $(y_1,y_2,...,y_n)$ on $T(m,h)$ where $n=\frac{m^{h+1}-1}{m-1}$ and for all $i\in \{ 1,2,\dots ,n\}$,
\begin{equation*}
y_i=
\begin{cases}
\alpha_2 = \frac{(m-1)(m^{h+1}-m^h+1)}{m^{h+2}-m^{h+1}+m^h-1} & \text{ if }v_i\text{ is the root,}\\
\beta_2 = \frac{m^h-1}{m^{h+2}-m^{h+1}+m^h-1} & \text{ if } v_i\text{ has depth }1\\
0 & \text{ if }v_i\text{ has depth }d>1.
\end{cases}
\end{equation*}
\end{definition}

The probabilities $\alpha_2$ and $\beta_2$ in the strategy  $\mu_2$ were obtained by solving
\begin{align}
\label{mu2alphaandbeta}
\begin{split}
Gain(T(m,h),Z(0,0),\mu_2)&=Gain(T(m,h),Z(1,1),\mu_2) \\
\Leftrightarrow \alpha_2 \cdot 0 + m \beta_2 \cdot m^h &= \alpha_2 \cdot \left( \frac{m^h-1}{m-1} \right) + \beta_2 \cdot 0 + (m-1) \beta_2 \left( \frac{m^h-1}{m-1} \right),3 
\end{split}
\end{align}
subject to $\alpha_2+m\beta_2=1$.  In both these scenarios, it makes sense to equal the expected gains since the players do not want the give an advantage to their opponent of choosing a given vertex over another on which they assign a positive probability.

Considering the strategy $\mu_1$ as a safe strategy for Player 1 on $T(m,h)$ and the strategy  $\mu_2$ as an opposing strategy for Player 2 on $T(m,h)$ leads to the following result.
\begin{theorem}
\label{MaryTreeSafetyValue}
In the two-player game of Competitive Diffusion on  $T(m,h)$, the safety value of Player 1 is 
\begin{equation*}
\frac{(n-1)((m-1)n+1)}{n(m^2-m+1)+m-1}
\end{equation*}
where $n=\frac{m^{h+1}-1}{m-1}$. Moreover, Player 1 achieves the greatest gain with the safe strategy $\mu_1$, and the best opposing strategy for Player 2 is strategy $\mu_2$.
\end{theorem}
\begin{proof}
As a lower bound, we have the guaranteed gain of Player 1 with the strategy $\mu_1$. It is determined similarly as for the spiders, by taking the minimum gain for Player 1 over all pure opposing strategies $Z(d,e)$
Due to symmetry, we only need to consider the root of the tree and one vertex of each level as possible starting vertices for Player 2. By (\ref{mu1alphaandbeta}) we have that
\begin{align}
Gain(T(m,h),\mu_1,Z(0,0))&= Gain(T(m,h),\mu_1,Z(1,e))\notag \\
&= \alpha_1 \cdot 0 + m \beta_1 \left( \frac{m^h-1}{m-1} \right)\notag \\
&=\frac{m^{h+1}(m^h-1)}{m^{h+2}-m^{h+1}+m^h-1}.
\label{Equation2}
\end{align}

If $2 \leq d \leq h$, the expected gain of Player 1 is greater than $m^h$ since \\ $Gain(T(m,h),Z(0,0), Z(d,e))$ and $Gain(T(m,h),Z(1,e'),Z(d,e))$ for $0<e'\leq m-1$ and $0 \leq e \leq m^d-1$ are both greater than $m^h$. Moreover, $m^h$ is greater than (\ref{Equation2}). Thus, $GGain(T(m,h),\mu_1)=Gain(T(m,h),\mu_1,Z(0,0))$. 

For the upper bound, we have the maximum gain of Player 1 when Player 2 chooses the strategy $\mu_2$. It is determined by taking the maximum gain over all pure strategies $Z(d,e)$ for Player 1. 
By (\ref{mu2alphaandbeta}) we have that
\begin{align*}
Gain(T(m,h),Z(0,0),\mu_2)&=Gain(T(m,h),Z(1,e),\mu_2)\\
&=\alpha_2 \cdot 0 + m \beta_2 \cdot m^h\\
&=\frac{m^{h+1}(m^h-1)}{m^{h+2}-m^{h+1}+m^h-1}.
\end{align*}
Furthermore, for $2 \leq d \leq h$,
\begin{equation*}
\begin{split}
Gain(T(m,h),Z(d,e),\mu_2) \leq (\alpha_2 + \beta_2) \left( \frac{m^{h-1}-1}{m-1} \right)
+ (m-1)\beta_2 \left(\frac{m^h-1}{m-1}\right).
\end{split}
\end{equation*}
Substituting the expressions for  $\alpha_2$ and $\beta_2$, we can show that this gain is smaller than  $Gain(T(m,h),Z(0,0),\mu_2)$. Thus, $MGain(T(m,h),\mu_2)=Gain(T(m,h),Z(0,0),\mu_2)$. 
Finally, we have 
\begin{align*}
GGain(T(m,h),\mu_1)=MGain(T(m,h),\mu_2)=\frac{(n-1)((m-1)n+1)}{n(m^2-m+1)+m-1}
\end{align*}
since $n=\frac{m^{h+1}-1}{m-1}$.
\end{proof}

\section{An Algorithm to Find Safe Strategies for Trees in General}

In this section we exploit our earlier results to develop a heuristic algorithm to find a good safe strategy for any tree. We assume our tree to be centroidal. Such a tree has the centroid as its root, and $n-1$ vertices divided amongst a number of branches. For bicentroidal trees, we can adopt a similar approach by considering one of the two vertices of the centroid to be the root. 

The branches extending from the centroid can have different configurations. In spiders, all the branches at the centroid are non-trivial paths. We showed that a safe strategy which chooses with positive probability vertices on the branches up to a certain distance has a guaranteed gain near the safety value. On the other hand, the branches in a complete tree are more clustered. We showed that a safe strategy which has a guaranteed gain equal to the safety value only chooses the root and the first vertex of each branch. This suggests considering different types of branches at the centroid and defining accordingly a distribution of probabilities on the vertices in the branch. 

\subsection{Branches at the Centroid}
We distinguish three different types of branches at the centroid.

\begin{definition}
A {\em thick branch} at the centroid is a branch for which we have 
\begin{equation*}
w_2 \geq n-w_1+\frac{w_1^2}{n}
\end{equation*} 
where $w_2$ is the second lowest weight in the branch and $w_1$ is the lowest weight in the branch.

A {\em medium branch} at the centroid is a branch for which we have 
\begin{equation*}
w_2 <n-w_1+\frac{w_1^2}{n} \text{ and } w_3 \geq n-w_2+\frac{w_2^2+(w_2-w_1)^2}{n+(w_2-w_1)}
\end{equation*} 
where $w_3$ is the third lowest weight in the branch, $w_2$ is the second lowest weight in the branch and $w_1$ is the lowest weight in the branch.

A {\em thin branch} at the centroid is a branch for which we have 
\begin{equation*}
w_2 <n-w_1+\frac{w_1^2}{n} \text{ and } w_3 < n-w_2+\frac{w_2^2+(w_2-w_1)^2}{n+(w_2-w_1)}
\end{equation*} 
where $w_3$ is the third lowest weight in the branch, $w_2$ is the second lowest weight in the branch and $w_1$ is the lowest weight in the branch.
\end{definition}

By considering that the weight of a vertex is the number of edges in the branch in which lies the centroid (see Lemma \ref{TreeCentroidinBranch}), one can show that the vertex with the lowest weight in a branch is adjacent to the centroid and that the vertex with the second lowest weight is adjacent to the first. The next vertex with the lowest weight in the branch could be adjacent to either of these vertices. However, condition (\ref{ConditionThinBranch}) assures us that the third vertex with the lowest weight in a thin branch is adjacent to the vertex with the second lowest weight. It would be impossible to have condition (\ref{ConditionThinBranch}) and the second and third vertices with lowest weights on two different branches at the vertex with the lowest weight, since this would imply $n-1=w_1+(n-w_2)+(n-w_3)$. 

Following an approach similar to that involved in finding the strategies $\mu_1$ and $\mu_2$, we give an algorithm which, given any centroidal tree $T$ of size $n$, assigns a distribution of probabilities on the vertices of a branch which depends on the type of the branch. Let $B_i$ be a branch  of $T$ and let $u_i$, $t_i$ and $s_i$ be the vertices in $B_i$ such that $w(u_i)\leq w(t_i)\leq w(s_i) \leq w(v_k)$ for any other vertex $v_k$ in $B_i$. As explained in the previous paragraph, $u_i$ is adjacent to the root, and $t_i$ is adjacent to $u_i$.  Vertex $s_i$ could be adjacent to $u_i$ or $t_i$, but in case of a thin branch, $s_i$ is adjacent to $t_i$. The algorithm proposed here will assign a probability $\alpha$ to the root (centroid), and, in each branch $B_i$, assigns probabilities $\beta_i$, $\gamma_i$ and $\delta_i$ to the vertices $u_i$, $t_i$ and $s_i$ and probability zero to all other vertices in $B_i$.

The probabilities $\beta_i$, $\gamma_i$ and $\delta_i$ are given below. The expressions are given in terms of the weights of vertices $u_i$, $t_i$ and $s_i$ and of  $\alpha$, the probability assigned to the root  vertex. The expression depends on whether the branch $B_i$ is a thin, medium, or thick branch. 

\begin{itemize}
\item[(i)] If $B_i$ is a thin branch,

\begin{align}
\label{ThinbranchDistribution} 
\begin{split}
&\beta_i=\left(\frac{\overline{w}(t_i)(w(u_i)\overline{w}(s_i)+(w(t_i)-w(s_i))(w(t_i)-w(u_i)))}{\overline{w}(s_i)\overline{w}(u_i)\overline{w}(t_i)+w(s_i)w(t_i)(w(s_i)-w(t_i))}\right) \alpha  \\
&\gamma_i = \left(\frac{w(t_i)}{\overline{w}(t_i)}\right) \beta_i  \\
&\delta_i = \left( \frac{w(s_i)}{\overline{w}(s_i)} \right) \gamma_i + \left( \frac{w(t_i)-w(u_i)}{\overline{w}(s_i)}\right) \alpha 
\end{split} 
\end{align}
\item[(ii)] If $B_i$ is a medium branch, 
\begin{align}
\label{MediumbranchDistribution}
\begin{split}
& \beta_i = \left(\frac{w(u_i)}{\overline{w}(u_i)}\right) \alpha  \\
&  \gamma_i = \left(\frac{w(t_i)}{\overline{w}(t_i)}\right) \beta_i \\
&  \text{ } \delta_i = 0 
\end{split}
\end{align}
\item[(iii)] If $B_i$ is a thick branch, 
\begin{align}
\label{ThickbranchDistribution}
\begin{split}
&  \beta_i=\left(\frac{w(u_i)}{\overline{w}(u_i)}\right) \alpha  \\
& \gamma_i = 0 \\
& \delta_i = 0 
\end{split} 
\end{align}
\end{itemize}

If $B_i$ is a thin branch, the suggested probabilities were obtained by equalling the expected gains of Player 1 when Player 2 chooses the centroid, the vertex $u_i$, the vertex $t_i$ and the vertex $s_i$ and solving for $\beta_i$, $\gamma_i$ and $\delta_i$ knowing that $\alpha+\beta_i+\gamma_i+\delta_i=1$. Similarly, the suggested probabilities if $B_i$ is a medium branch were obtained by first setting $\delta_i$ to zero, then equalling the expected gains of Player 1 when Player 2 chooses the centroid, the vertex $u_i$ and the vertex $t_i$. Finally, if $B_i$ is a thick branch, we set $\gamma_i$ and $\delta_i$ to zero and we equal the expected gains of Player 1 when Player 2 chooses the centroid and the vertex $u_i$. The distribution of probabilities in the branches will be used in the suggested safe strategy for Player 1.

\subsection{Centroidal Safe Strategy (CSS) Algorithm}
A centroidal tree can have diverse proportions of thin, medium and thick branches. Moreover, unlike spiders and complete trees, the number of vertices varies from branch to branch. Thus, some branches might have very few vertices compared to other branches. Therefore, it would be unreasonable to suggest a safe strategy for Player 1 which has positive probability of choosing vertices in every branch at the centroid. As such, we define an algorithm to suggest a safe strategy for Player 1. The algorithm starts by assigning positive probabilities to vertices on one branch and then disperses probabilities on other branches as long as it is beneficial to Player 1. To determine whether the assignment is beneficial, we define the criterion of a branch. This criterion will be used in the algorithm to order the branches and determine when the dispersion of probabilities on branches should stop.

\begin{definition}
\label{DefCriterionBranch}
For a branch $B$ in a centroidal tree, we define the {\em criterion} of $B$, $Cr(B)$, as follows. If $B$ has less than three vertices, $Cr(B)=0$. Else, let $u$, $t$ and $s$ be the vertices in $B$ such that $w(u)\leq w(t)\leq w(s) \leq w(v)$ for any other vertex $v$ in $B$, and $u$ is a neighbour of the centroid, $t$ a neighbour of $u$, and $s$ a neighbour of $t$ or $u$. Then,
\end{definition}
\begin{align}
\begin{split}
&Cr(B)=\\ &\begin{cases}
\overline{w}(u) & \text{ if } B \text{ is a thick branch,} \\
\left(\frac{\overline{w}(t)}{n}\right)\overline{w}(u)+\left(\frac{w(t)}{n}\right)\overline{w}(t) & \text{ if } B \text{ is a medium branch,}\\
\frac{w(t)\overline{w}(t)(n^2-nw(s)-w(s)w(t)+w(t)^2+2w(s)w(u)-w(t)w(u))}{nw(t)\overline{w}(s)+w(u)w(t)(-n+w(s)+w(t))+\overline{w}(t)w(u)^2} & \text{ if } B \text{ is a thin branch.}
\end{cases}
\end{split}
\end{align}

\begin{algorithm} Centroidal Safe Strategy (CSS) Algorithm
\noindent \par\nobreak \hrule\vspace{6pt}
\label{CSSAlgorithm}
\begin{itemize}
\item[] \textsc{INPUT:} Centroidal tree, $T$, with $d$ branches at the centroid.
\item[] \textsc{STEP 1:} Order the branches $\{B_1,B_2,...,B_d\}$ such that $Cr(B_i)\geq Cr(B_{i+1})$ for all $1 \leq i \leq d-1$.

\item[] \textsc{STEP 2:} Build a sequence of safe strategies $\sigma_i$ for Player 1 by considering each branch in order.
\begin{itemize}
\item[(a)] If $i=0$, form $\sigma_0$, a strategy where  the centroid is chosen with probability $\alpha$. 

 %and consider the branch $B_1$. Form $\sigma_1$, a safe strategy for Player 1 in which the centroid is chosen with probability $\alpha$ and the probabilities of choosing the vertices in the branch $B_1$ are (\ref{ThinbranchDistribution}), if $B_1$ is a thin branch, (\ref{MediumbranchDistribution}), if $B_1$ is a medium branch and (\ref{ThickbranchDistribution}), if $B_1$ is a thick branch.
%Skip to STEP 3.

\item[(b)] If $i>0$, form, $\sigma_i$, a safe strategy in which the centroid is chosen with probability $\alpha$, the probabilities of choosing the vertices of the branches $B_k$, $1\leq k<i$ are the same in terms of $\alpha$ as in the strategy $\sigma_{i-1}$ and the probabilities of choosing the vertices in the branch $B_i$ are as given by (\ref{ThinbranchDistribution}), if $B_i$ is a thin branch, by (\ref{MediumbranchDistribution}), if $B_i$ is a medium branch and by (\ref{ThickbranchDistribution}),  if $B_i$ is a thick branch.
\end{itemize}

\item[] \textsc{STEP 3:} Determine $\alpha$ by solving $\alpha +\sum_{j=1}^i \left(\beta_j+ \gamma_j + \delta_j \right) =1$ and calculate the expected gain of Player 1 with the strategy $\sigma_i$ when Player 2 chooses the centroid,
\begin{equation}
\label{Equation5}
Gain(T,\sigma_i,Z(c))=\alpha \cdot 0+ \sum_{j=1}^i \left(\beta_j \cdot \overline{w}(u_j)+\gamma_j \cdot \overline{w}(t_j)+\delta_j \cdot \overline{w}(t_j)\right).
\end{equation}
\item[] \textsc{STEP 4:}
\begin{itemize}
\item[(a)] If $i<d$ and $Cr(B_{i+1}) \geq Gain(T,\sigma_i,Z(c))$ return to \textsc{STEP 2} with $i=i+1$.
\item[(b)] If $i<d$ and $Cr(B_{i+1}) \leq Gain(T(n),\sigma_i,Z(c))$ or $i=d$, return the strategy $\sigma_i$ and the guaranteed gain, $Gain(T,\sigma_i,Z(c))$.
\end{itemize} 
\item[] \textsc{OUTPUT:} Safe strategy for Player 1, $\sigma_i$, with guaranteed gain
\begin{equation}
GGain(T,\sigma_i)=Gain(T,\sigma_i,Z(c)).
\end{equation}
\end{itemize}
\noindent \par\nobreak \hrule\vspace{6pt}
\end{algorithm}
A few explanations on the algorithm are needed. In \textsc{Step} 4, we return to \textsc{Step 2} to disperse probabilities on another branch if $Cr(B_{i+1})\geq Gain(T,\sigma_i,Z(c))$. The criterion being greater than the current expected gain results in an increase of the expected gain of Player 1. Thus, the strategies $\sigma_i$ give increased gain. This is shown in the following lemma. 

\begin{lemma}
\label{LemmaConditionCriterion}
In the CSS algorithm, if $Cr(B_{i+1})\geq Gain(T,\sigma_i,Z(c))$, then 
\begin{align}
Gain(T,\sigma_i,Z(c)) \leq Gain(T,\sigma_{i+1},Z(c))\leq Cr(B_{i+1}).
\end{align}
\end{lemma}

\begin{proof}
Suppose $B_{i+1}$ is thick branch. If $B_{i+1}$ is a medium or thin branch, the proof is similar. We refer here to  the solved value of $\alpha$ in the strategy $\sigma_i$ by $\alpha^{(i)}$, and let strategy $\sigma_i$ be represented by the vector $\alpha^{(i)}(x_1^{(i)}, x_2^{(i)}, ..., x_n^{(i)})$. Let $v_k$ be the index of the vertex $u_{i+1}$ of $B_{i+1}$.% and let the strategy $\sigma_{i+1}$ be represented by the vector $\alpha^{(i+1)}(x_1^{(i+1)}, x_2^{(i+1)}, ..., x_n^{(i+1)})$.  
To form the strategy $\sigma_{i+1}$, a probability of $\beta_{j+1}=\left(\frac{w(u_{i+1})}{\overline{w}(u_{i+1})}\right) \alpha$  was assigned to the vertex $v_k$ and we have a new solved value for $\alpha$, $\alpha^{(i+1)}$. Moreover, we have the following relation between the probability vectors. 
\begin{align*}
x_j^{(i+1)}=
\begin{cases}
x_j^{(i)} & \text{if } j \not = k\\
\frac{w(u_{i+1})}{\overline{w}(u_{i+1})}, & \text{if } j=k,
\end{cases}
\end{align*}
and
\begin{equation}
\label{eqnalpha}
\frac{1}{\alpha^{(i+1)}}=\sum_{j=1}^n x_j^{(i+1)}=\frac{1}{\alpha^{(i)}}+\frac{w(u_{i+1})}{\overline{w}(u_{i+1})}.
\end{equation}
If we compare the expressions for $Gain(T,\sigma_i,Z(c))$ and $Gain(T,\sigma_{i+1},Z(c))$, we have 
\begin{align*}
\begin{split}
Gain(T,\sigma_{i+1},Z(c))&=\sum_{j=1}^n \alpha^{(i+1)} x_j^{(i+1)} \cdot Gain(T,Z(v_j),Z(c))\\
&=\sum_{j=1,j\not= k}^n \alpha^{(i+1)} x_j^{(i)} Gain(T,Z(v_j),Z(c))\\&+\left(\frac{w(u_{i+1})}{\overline{w}(u_{i+1})}\right) \alpha^{(i+1)}  Gain(T,Z(v_k),Z(c)).
\end{split}
\end{align*}
Now, $x_k^{(i)}=0$, since initially, there were no positive probabilities on the vertices of the branch $B_{i+1}$. 
Remembering that $v_k=u_{i+1}$, we obtain 
\begin{align}
\label{GainEqn}
\begin{split}
Gain(T,\sigma_{i+1},Z(c))&=\frac{\alpha^{(i+1)}}{ \alpha^{(i)}} Gain(T,\sigma_i,Z(c)) \\
&+ \left(\frac{w(u_{i+1})}{\overline{w}(u_{i+1})}\right)\alpha^{(i+1)}   Gain(T,Z(u_{i+1}),Z(c)) \\
&=\nu \,Gain(T,\sigma_i,Z(c)) 
+ (1-\nu) Gain(T,Z(u_{i+1}),Z(c)) ,
\end{split}
\end{align}
where $\nu =\frac{\alpha^{(i+1)}}{ \alpha^{(i)}}$, and the last step uses (\ref{eqnalpha}). 

From Step 4 of the algorithm, we know that, if the algorithm continues to form $\sigma_{i+1}$, then $Cr(B_{i+1}) \geq Gain(T,\sigma_{i},Z(c))$. 
%In this case, $i$ is increased to $i+1$, so we also have that $Cr(B_{i+1}) \geq Gain(T,\sigma_{i+1},Z(c))$. This proves the last part of the statement of the lemma.
Since $B_{i+1}$ is a thick branch, $Cr(B_{i+1})=\overline{w}(u_{i+1})$ (see Definition \ref{DefCriterionBranch}). Recall from Lemma \ref{TreeCentroidinBranch}, that $w(u_{i+1})$ is the number of edges in the branch at $u_{i+1}$ in which lies the centroid. Thus, there are $\overline{w}(u_{i+1})$ vertices in the branch $B_{i+1}$ and so 
\[
Gain(T,Z(u_{i+1}),Z(c))=\overline{w}(u_{i+1})=Cr(B_{i+1})\geq Gain(T,\sigma_{i},Z(c)).
\]

By  (\ref{GainEqn}), $Gain(T,\sigma_{i+1},Z(c))$ is a convex combination of $Gain(T,Z(u_{i+1}),Z(c))$ and $Gain(T,\sigma_i,Z(c)) $. Thus the inequalities stated in the lemma follow. 

\end{proof}

The lemma allows us to explain the ordering of the branches. After $k$ loops in the algorithm, suppose we have two branches at the centroid, $B_i$ and $B_j$ such that $Cr(B_i)\geq Cr(B_j) \geq Gain(T,\sigma_k,Z(c))$. By Lemma \ref{LemmaConditionCriterion}, we know that adding positive probabilities on either of the branches $B_i$ and $B_j$ will increase the expected gain of Player 1. Moreover, the resulting expected gains will not surpass $Cr(B_i)$ and $Cr(B_j)$ respectively. If we start by adding the probabilities on the branch $B_j$, the branch $B_i$ remains a candidate to increase the expected gain once more, since $Cr(B_i)\geq Cr(B_j)$. Let $\nu$ be the resulting expected gain of Player 1 with the two branches added. On the other hand, if we start with  branch $B_i$, the resulting expected gain might or might not be lower than $Cr(B_j)$. If it is, we can add the branch $B_j$ to get the expected gain $\nu$. If it is not, then the resulting expected gain with only the branch $B_i$ is greater than $\nu$. Thus, it is always advantageous to include the branch with the largest criterion first. For this reason, we order the branches in decreasing order of criterion in the algorithm. 

Lastly, we must show that the minimal expected gain of Player 1 with the strategy $\sigma_i$ is in obtained when Player 2 chooses the centroid.
\begin{theorem}
\label{ThmTreesAlgorithmGGain}
Let $T$ be a centroidal tree of size $n$ with $d$ branches at the centroid. Suppose we apply the Safe Strategy Algorithm to $T$ and we get the mixed strategy $\sigma_k$ of Player 1 as output. Then, 
\begin{equation*}
GGain(T,\sigma_k)=Gain(T,\sigma_k,Z(c))
\end{equation*}
where $c$ is the centroid of $T
$.
\end{theorem}
\begin{proof} 
Suppose $T$ has $k_1$ thick branches, $\{i^{Thk}_1,i^{Thk}_2,...,i^{Thk}_{k_1}\} \in \{1,2,3,...,k\}$, $k_2$ medium branches,
$\{i^{Med}_1,i^{Med}_2,...,i^{Med}_{k_2}\} \in \{1,2,3,...,k\}$ and $k_3$ thin branches,
$\{i^{Thn}_1,i^{Thn}_2,...,i^{Thn}_{k_3}\}$ $ \in \{1,2,3,...,k\}$ from the set of branches $\{B_1,B_2,...,B_k\}$, $k_1+k_2+k_3=k$.

Proving this theorem consists of determining the expected gain of Player 1 over all the possible starting vertices for Player 2 and showing that the minimum occurs when Player 1 chooses the centroid.  If Player 2 chooses to start with the centroid, the expected gain of Player 1 with the strategy $\sigma_k$ is 
\begin{align}
\label{ProofTrees7}
\begin{split}
Gain(T,\sigma_k,Z(c)) &=\alpha \cdot 0 + \sum_{j=i^{Thk}_1}^{i^{Thk}_{k_1}} \beta_j\cdot\overline{w}(u_j)+\sum_{j=i^{Med}_1}^{i^{Med}_{k_2}} \left( \beta_j \cdot\overline{w}(u_j) + \gamma_j \cdot\overline{w}(t_j) \right) \\
&+\sum_{j=i^{Thn}_1}^{i^{Thn}_{k_3}} \left( \beta_j \cdot\overline{w}(u_j)+ \gamma_j \cdot\overline{w}(t_j)+ \delta_j \cdot\overline{w}(t_j)\right).
\end{split}
\end{align}

The gain for all other pure strategies of Player 2 can be determined from the definitions given earlier. The details of the proof are highly technical, and can be found in the Appendix. 
\end{proof}

\section{Experimental Assessment of the CSS Algorithm}

As a last section, we apply the CSS algorithm on some examples of centroidal trees in order to evaluate its guaranteed gain.  We generated random centroidal trees with $n=100$ and $n=1000$ vertices. To evaluate the proximity of the guaranteed gain to the safety value, we calculate its difference to the maximal gain of Player 1 against a strategy for Player 2 which chooses with positive probabilities the centroid and some of the vertices at distance 1 and 2 from the centroid. As explained in the first section, any opposing strategy for Player 2 gives an upper bound on the safety value. 

The trees were generated using Maple$^{\text{\texttrademark}}$ \cite{MAPLE16} and the computation of the CSS algorithm was carried out using MATLAB$^{\text{\textregistered}}$ \cite{MATLAB2012}. The Maple$^{\text{\texttrademark}}$ function used for the generation of the trees is \texttt{RandomTree(n)} which has a randomized process as follows: ``Starting with the empty undirected graph $T$ on $n$ vertices, edges are chosen uniformly at random and inserted into $T$ if they do do not create a cycle.  This is repeated until $T$ has $n-1$ edges." \cite{MAPLE16ManualRandomTree} \textsc{Figure} \ref{TreesExamplesRandom} shows the number of examples with a difference between guaranteed gain obtained with the strategy given by the CSS algorithm, and the best upper bound. The difference is given as a proportion of the weight of the centroid. The columns represent the number of examples, out of a total of 1000, with a difference in the intervals $[0,0]$, $(0,0.01]$, $(0.01,0.02]$,...,$(0.29,0.30]$ respectively. 
\begin{figure}
   \centering
   \subfigure[Frequency Graph with $n=100$.]{\includegraphics[width=2.81 in]{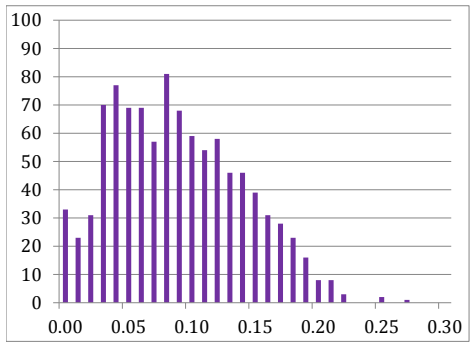}} \qquad
 %  \subfigure[Cumulative Frequency Graph with $n=100$.]{\includegraphics[width=2.81 in]{Treesn=100(2)}}\qquad
   \subfigure[Frequency Graph with $n=1000$.]{\includegraphics[width=2.81 in]{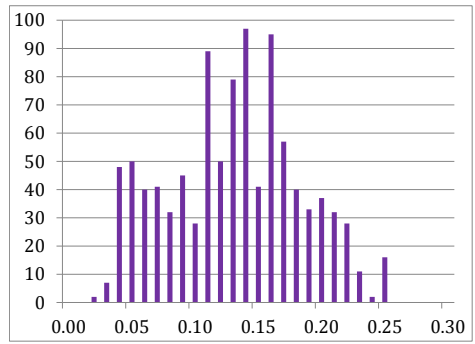}}\qquad
  % \subfigure[Cumulative Frequency Graph with $n=1000$.]{\includegraphics[width=2.81 in]{Treesn=1000(2)}}\qquad
   \caption{\label{TreesExamplesRandom} Frequency of the differences between the guaranteed gain of the CSS Algorithm and the maximal gain of Player 1 as a proportion of the weight of the centroid in 1 000 random centroidal tree examples with total number of vertices, $n$.}
 \end{figure}

We see that the algorithm performs well and even ideally on a number of examples.  
 A weaker performance was observed mostly for trees with large thin branches. This can be explained from our result for spiders, where all branches are thin. For spiders, the best strategy assigns positive probabilities to vertices at larger distance (about square root of the length) from the root. The CSS algorithm limits the number of vertices with positive probability in any branch to three. Thus, a trade-off of the versatility of the CSS algorithm is that the performance on thin branches is not optimal. 

\section{Conclusion and Future Work}

We have explored the safe game for competitive diffusion for trees. We obtained precise results for special classes of trees, namely spiders and complete trees. These results were then incorporated in the CSS algorithm, which can be applied to any tree. This algorithm was evaluated experimentally, and was shown to give good results. 

However, performance of the CSS algorithm decreased with the presence of many thin branches. Generalizing the ideas presented here to include positive probabilities on more vertices on the branches would improve the algorithm. One might also consider slightly modifying the ordering of the branches or the distribution of the probabilities to compensate for the branches being considered isolated when the suggested distribution of probabilities were calculated. 

We believe that the ideas put forward in this paper can be extended beyond trees. Namely, the spread of influence can be seen as taking place in the form of a subtree of the graph. The general approach, of assigning positive probabilities only to relatively few vertices close to the center of the graph, is likely to be of value here as well. Thus, it seems plausable that the CSS Algorithm might be modified to a more general setting. It may not be possible to obtain tight bounds for the safety value in the more general setting, but it should be possible to generate safe strategies that perform well in practice. 

\bibliography{Bibliography}

\begin{thebibliography}{10}

\bibitem{Alon2010}
Noga Alon, Michal Feldman, Ariel~D. Procaccia, and Moshe Tennenholtz.
\newblock A note on competitive diffusion through social networks.
\newblock {\em Information Processing Letters}, 110(6):221 -- 225, 2010.

\bibitem{Sayan2013}
Sayan Bandyapadhyay, Aritra Banik, Sandip Das, and Hirak Sarkar.
\newblock Voronoi game on graphs.
\newblock In {\em WALCOM: Algorithms and Computation}, volume 7748 of {\em
  Lecture Notes in Computer Science}, pages 77--88. Springer Berlin Heidelberg,
  2013.

\bibitem{Borodin2010}
Allan Borodin, Yuval Filmus, and Joel Oren.
\newblock Threshold models for competitive influence in social networks.
\newblock In {\em Internet and Network Economics}, volume 6484 of {\em Lecture
  Notes in Computer Science}, pages 539--550. Springer Berlin Heidelberg, 2010.

\bibitem{Boudreau2013}
Daniel Boudreau, Jeannette Janssen, Richard Nowakowski, and Elham Roshanbin.
\newblock Safe strategies for competitive diffusion in social networks.
\newblock Manuscript.

\bibitem{Goldenberg2001}
Jacob Goldenberg, Barak Libai, and Eitan Muller.
\newblock Talk of the network: A complex systems look at the underlying process
  of word-of-mouth.
\newblock {\em Marketing Letters}, 12(3):211--223, 2001.

\bibitem{Goyal2012}
Sanjeev Goyal and Michael Kearns.
\newblock Competitive contagion in networks.
\newblock In {\em Proceedings of the 44th Symposium on Theory of Computing},
  STOC '12, pages 759--774, New York, NY, USA, 2012. ACM.

\bibitem{Granovetter1978}
Mark Granovetter.
\newblock Threshold models of collective behavior.
\newblock {\em American Journal of Sociology}, 83(6):pp. 1420--1443, 1978.

\bibitem{Kang1975}
Andy~N.C. Kang and David~A. Ault.
\newblock Some properties of a centroid of a free tree.
\newblock {\em Information Processing Letters}, 4(1):18 -- 20, 1975.

\bibitem{Kempe2003}
David Kempe, Jon Kleinberg, and \'{E}va Tardos.
\newblock Maximizing the spread of influence through a social network.
\newblock In {\em Proceedings of the Ninth ACM SIGKDD International Conference
  on Knowledge Discovery and Data Mining}, KDD '03, pages 137--146, New York,
  NY, USA, 2003. ACM.

\bibitem{Knuth}
Donald~E. Knuth.
\newblock {\em The Art of Computer Programming}.
\newblock Addison-Wesley, Reading, Mass, 1997.

\bibitem{MAPLE16}
{\em {Maple 16}}.
\newblock Maplesoft, a division of Waterloo Maple Inc., Waterloo, Ontario,
  2012.

\bibitem{MAPLE16ManualRandomTree}
{Maplesoft a division of Waterloo Maple Inc. Maple User Manual.
  GraphTheory[RandomGraphs][RandomTree]}, (2005-2013).
\newblock
  \url{http://www.maplesoft.com/support/help/Maple/view.aspx?path=GraphTheory/RandomGraphs/RandomTree},
  Date accessed: January 23rd, 2014.

\bibitem{MATLAB2012}
{\em {MATLAB and Statistics Toolbox Version 8.0.0.783 (R2012b)}}.
\newblock The MathWorks, Inc., Natick, Massachusetts, 2012.

\bibitem{Mitchell1978}
Sandra~L. Mitchell.
\newblock Another characterization of the centroid of a tree.
\newblock {\em Discrete Mathematics}, 24(3):277 -- 280, 1978.

\bibitem{Chen2009}
Chen Ning.
\newblock On the approximability of influence in social networks.
\newblock {\em SIAM Journal on Discrete Mathematics}, 23(3):1400 -- 1415, 2009.

\bibitem{Roshanbin2014}
Elham Roshanbin.
\newblock The competitive diffusion game in classes of graphs.
\newblock In {\em Proceedings of 10th International conference on Algrotithmic
  Aspects of Information and Management (AAIM)}, Vancouver, 2014.

\bibitem{Lin2011}
Jen-Ling Shang and Chiang Lin.
\newblock Spiders are status unique in trees.
\newblock {\em Discrete Mathematics}, 311(10–11):785 -- 791, 2011.

\bibitem{Small2013}
Lucy Small and Oliver Mason.
\newblock Information diffusion on the iterated local transitivity model of
  online social networks.
\newblock {\em Discrete Applied Mathematics}, 161(10–11):1338 -- 1344, 2013.

\bibitem{Small20132}
Lucy Small and Oliver Mason.
\newblock Nash equilibria for competitive information diffusion on trees.
\newblock {\em Information Processing Letters}, 113(7):217 -- 219, 2013.

\bibitem{Takehara2012}
Reiko Takehara, Masahiro Hachimori, and Maiko Shigeno.
\newblock A comment on pure-strategy nash equilibria in competitive diffusion
  games.
\newblock {\em Information Processing Letters}, 112(3):59 -- 60, 2012.

\bibitem{Demaine2011}
{Sachio} {Teramoto}, {Erik}~D. {Demaine}, and {Ryuhei} {Uehara}.
\newblock The voronoi game on graphs and its complexity.
\newblock {\em Journal of Graph Algorithms and Applications}, 15(4):485--501,
  2011.

\bibitem{Goldberg2012}
Vasileios Tzoumas, Christos Amanatidis, and Evangelos Markakis.
\newblock A game-theoretic analysis of a competitive diffusion process over
  social networks.
\newblock In {\em Internet and Network Economics}, volume 7695 of {\em Lecture
  Notes in Computer Science}, pages 1--14. Springer Berlin Heidelberg, 2012.

\end{thebibliography}
\bibliographystyle{plain}

\section*{Appendix: Proof of Theorem \ref{ThmTreesAlgorithmGGain}}

\begin{proof}
%Suppose $T$ has $k_1$ thick branches, $\{i^{Thk}_1,i^{Thk}_2,...,i^{Thk}_{k_1}\} \in \{1,2,3,...,k\}$,
%$k_2$ medium branches,
%$\{i^{Med}_1,i^{Med}_2,...,i^{Med}_{k_2}\} \in \{1,2,3,...,k\}$,
%and $k_3$ thin branches,
%$\{i^{Thn}_1,i^{Thn}_2,...,i^{Thn}_{k_3}\} \in \{1,2,3,...,k\}$ from the set of branches $\{B_1,B_2,...,B_k\}$, $k_1+k_2+k_3=k$. 
%We must determine the expected gain of Player 1 over all the possible starting vertices for Player 2.

If Player 2 chooses to start with the centroid $c$, the expected gain of Player 1 with the strategy $\sigma_k$ is as given earlier in (\ref{ProofTrees7}). The labelling of the branches is as given before (\ref{ProofTrees7})
%\begin{align}
%\label{ProofTrees77}
%\begin{split}
%Gain(T,\sigma_k,Z(c)) &=\alpha \cdot 0 + \sum_{j=i^{Thk}_1}^{i^{Thk}_{k_1}} \beta_j\cdot\overline{w}(u_j)+\sum_{j=i^{Med}_1}^{i^{Med}_{k_2}} \left( \beta_j \cdot\overline{w}(u_j) + \gamma_j \cdot\overline{w}(t_j) \right) \\
%&+\sum_{j=i^{Thn}_1}^{i^{Thn}_{k_3}} \left( \beta_j \cdot\overline{w}(u_j)+ \gamma_j \cdot\overline{w}(t_j)+ \delta_j \cdot\overline{w}(t_j)\right).
%\end{split}
%\end{align}

Let us now consider the cases when Player 2 chooses a vertex in a thin branch, $B_r$, $r \in \{i^{Thn}_1,i^{Thn}_2,...,i^{Thn}_{k_3}\}$.
\begin{itemize}
\item[i)] If Player 2 chooses to start with the vertex $u_r$, the expected gain of Player 1 with the strategy $\sigma_k$ is
\begin{align}
\label{ProofTrees6}
\begin{split}
Gain(T,\sigma_k,Z(u_r)) &=\alpha \cdot w(u_r) +  \sum_{j=i^{Thk}_1}^{i^{Thk}_{k_1}} \beta_j \cdot \overline{w}(u_j)\\
&+\sum_{j=i^{Med}_1}^{i^{Med}_{k_2}} \left( \beta_j \cdot\overline{w}(u_j) + \gamma_j \cdot\overline{w}(u_j) \right)\\&+\beta_r \cdot 0+\gamma_r\cdot\overline{w}(t_r)+\delta_r\cdot\overline{w}(s_r)\\&+\sum_{j=i^{Thn}_1,\text{ } j\not = r}^{i^{Thn}_{k_2}} \left( \beta_j \cdot\overline{w}(u_j)+ \gamma_j \cdot\overline{w}(u_j)+ \delta_j \cdot\overline{w}(t_j)\right).
\end{split}
\end{align}
For all $j \in \{1,2,...,k\}$, by definition we have that $w(u_j)<w(t_j)$, and thus
\begin{equation}
\label{ProofTrees1}
\overline{w}(u_j)>\overline{w}(t_j).
\end{equation}
Also, by the definition of $\beta_r$,
\begin{equation}
\label{Equation15}
\beta_r\cdot\overline{w}(u_r)+\delta_r\cdot\overline{w}(t_r)=\alpha\cdot w(u_r)+\delta_r\cdot\overline{w}(s_r).
\end{equation}
 Using these, we can show that the expected gain (\ref{ProofTrees6}) is greater than or equal to $Gain(T,\sigma_k,c)$ of (\ref{ProofTrees7}).

\item[ii)] If Player 2 chooses vertex $t_r$, the expected gain of Player 1  is
\begin{align}
\label{ProofTrees8}
\begin{split}
Gain(T,\sigma_k,Z(t_r)) &=\alpha \cdot w(u_r) +  \sum_{j=i^{Thk}_1}^{i^{Thk}_{k_1}} \beta_j \cdot w(u_r)\\
&+\sum_{j=i^{Med}_1}^{i^{Med}_{k_2}} \left( \beta_j \cdot w(u_r) + \gamma_j \cdot \overline{w}(u_j) \right) \\
&+\beta_r \cdot w(t_r)+\gamma_r \cdot 0+ \delta_r\cdot \overline{w}(s_r)\\
&+\sum_{j=i^{Thn}_1,\text{ } j\not = r}^{i^{Thn}_{k_3}} \left( \beta_j \cdot w(u_r)+ \gamma_j \cdot\overline{w}(u_j)+ \delta_j \cdot\overline{w}(u_j)\right).
\end{split}
\end{align}
For all $j \in \{ 1,2,...k\}$, we have that 
\begin{equation}
\label{ProofTrees2}
w(u_r)>\overline{w}(u_j).
\end{equation} 
IF $j\not= r$, this follows since the branch at $u_r$ in which the centroid is located includes the edges in the branch  in which $u_j$ is located. If $j=r$, we have the result by Lemma \ref{Lemmanminusweight}. Moreover,
\begin{equation*}
\beta_r\cdot w(t_r)=\gamma_r \cdot \overline{w}(t_r)
\end{equation*}
by the definition of $\gamma_r$. 
Thus, the expected gain (\ref{ProofTrees8}) is greater than or equal to $Gain(T,\sigma_k,u_r)$ of (\ref{ProofTrees7}).

\item[iii)] If Player 2 chooses vertex $s_r$, the expected gain of Player 1  is
\begin{align*}
\begin{split}
Gain(T,\sigma_k,Z(s_r)) &=\alpha \cdot w(t_r) +  \sum_{j=i^{Thk}_1}^{i^{Thk}_{k_1}} \beta_j \cdot w(u_r)\\
&+\sum_{j=i^{Med}_1}^{i^{Med}_{k_2}} \left( \beta_j \cdot w(u_r) + \gamma_j \cdot w(u_r) \right) \\
&+\beta_r \cdot w(t_r)+\gamma_r \cdot w(s_r)+ \delta_r \cdot 0\\
&+\sum_{j=i^{Thn}_1,\text{ } j\not = r}^{i^{Thn}_{k_3}} \left( \beta_j \cdot w(u_r)+ \gamma_j \cdot w(u_r)+ \delta_j \cdot \overline{w}(u_j)\right).
\end{split}
\end{align*}
This expected gain is greater than or equal to $Gain(T,\sigma_k,t_r)$ by (\ref{ProofTrees2}) and  since
\begin{equation*}
\alpha \cdot w(u_r)+ \delta_r \cdot \overline{w}(s_r)=\alpha\cdot w(t_r)+\gamma \cdot w(s_r)
\end{equation*}
by the definition of $\delta_r$.
 
\item[iv)] If Player 2 chooses to start with a vertex $v_j$, $v_j \not\in\{ u_r,t_r,s_r\}$, the payoff to Player 1 on all vertices not part of the branch $B_r$ can only increase since Player 2's starting vertex is at a greater distance. Specifically, the payoff to Player 1 on the centroid is now at least $w(u_r)$. Moreover,
\begin{itemize}
\item[-] If $v_j$ is a descendant of $u_r$ but not of $t_r$ and $s_r$ then
\begin{align*}
\begin{split}
Gain(T,Z(u_r),Z(v_j))&\geq w(u_r) \\
Gain(T,Z(t_r),Z(v_j))&\geq \overline{w}(t_r) \\
Gain(T,Z(s_r),Z(v_j))&\geq \overline{w}(s_r).
\end{split}
\end{align*} 
\item[-] If $v_j$ is a descendant of $u_r$ and $t_r$ but not of $s_r$ then 
\begin{align*}
\begin{split}
Gain(T,Z(u_r),Z(v_j))&\geq w(u_r) \\
Gain(T,Z(t_r),Z(v_j))&\geq w(t_r) \\
Gain(T,Z(s_r),Z(v_j))&\geq \overline{w}(s_r).
\end{split}
\end{align*} 
\item[-] If $v_j$ is a descendant of $u_r$, $t_r$ and $s_r$, then
\begin{align*}
\begin{split}
Gain(T,Z(u_r),Z(v_j))&\geq w(u_r) \\
Gain(T,Z(t_r),Z(v_j))&\geq w(t_r) \\
Gain(T,Z(s_r),Z(v_j))&\geq w(s_r).
\end{split}
\end{align*}
\end{itemize}
Since $w(v)>\overline{w}(v)$ for any vertex $v$ other than the centroid by Lemma \ref{Lemmanminusweight}, in all cases we have 
\begin{align*}
\begin{split}
&\alpha \cdot Gain(T,Z(c),Z(v_j))+\beta_r \cdot Gain(T,Z(u_r),Z(v_j))\\ &+\gamma_r \cdot  Gain(T,Z(t_r),Z(v_j))+ \delta_r \cdot Gain(T,Z(s_r),Z(v_j)) \\ 
&\geq \alpha \cdot w(u_r)+ \beta_r \cdot w(u_r)+ \gamma_r \cdot \overline{w} (t_r)+\delta_r \cdot \overline{w} (s_r) \\
&> \alpha \cdot w(u_r)+  \gamma_r \cdot \overline{w}(t_r)+\delta_r \cdot \overline{w}(s_r)
\\ &= \beta_r \cdot \overline{w} (u_r)+\gamma_r \cdot \overline{w} (t_r)+\delta_r \cdot \overline{w} (t_r),
\end{split}
\end{align*} 
where the last equality follows from  (\ref{Equation15}). Now, 
\begin{align*}
\begin{split}
Gain(T,Z(c),Z(c))&=0,\\ Gain(T,Z(u_r),Z(c))&=\overline{w}(u_r), \\ Gain(T,Z(t_r),Z(c))&=\overline{w}(t_r) \text{ and } \\ Gain(T,Z(s_r),Z(c))&=\overline{w}(t_r),
\end{split}
\end{align*}
Thus,
\begin{align*}
\begin{split}
&\beta_r \cdot \overline{w} (u_r)+\gamma_r \cdot \overline{w} (t_r)+\delta_r \cdot \overline{w} (t_r)\\ &=\alpha \cdot Gain(T,Z(c),Z(c))+\beta_r \cdot Gain(T,Z(u_r),Z(c))\\ &+ \gamma_r \cdot Gain(T,Z(t_r),Z(c))+ \delta_r \cdot Gain(T,Z(s_r),Z(c)).
\end{split}
\end{align*}
Therefore, the expected gain of Player 1 when Player 2 chooses the vertex $v_j$ is greater than or equal to the expected gain of Player 1 when Player 2 chooses the centroid. 

Similarly, we can show that the expected gain of Player 1 when Player 2 chooses to start with a vertex in a medium branch or thick branch, is greater than the expected gain of Player 1 when Player 2 chooses to start with the centroid.

If Player 2 chooses  a vertex in a branch $B_i$, $i>k$ instead of the centroid, Player 1's payoff on the vertices in the branches $\{B_1,B_2,...,B_k\}$ can only increase. Namely, in this case strategy $\sigma_k$ assign no positive probabilities to any vertex in the branch. So, compared to the centroid, Player 2's starting vertex is at a greater distance from the vertices on which Player 1 has positive probability. Player 1's payoff on the centroid, being zero when Player 2 chooses the centroid, also increases. Thus, the expected gain of Player 1 is again greater.

To sum up, the expected gain of Player 1 with the strategy $\sigma_k$ is minimal when Player 2 chooses the centroid. Thus, $Gain(T,\sigma_k,Z(c))$ is the guaranteed gain of Player 1 with the strategy $\sigma_k$. 
\end{itemize}
\end{proof}
%\end{appendix}

%\thebibliography[heading=none]

\end{document}